\documentclass[twoside,10pt,reqno]{amsart}
\usepackage[T1]{fontenc}
\usepackage[utf8]{inputenc}
\usepackage{lmodern}
\usepackage{amsmath,amsfonts,amsthm,amssymb, enumerate}
\newtheorem{thm}{Theorem}[section]

\newtheorem{lem}[thm]{Lemma}

\theoremstyle{definition}
\newtheorem{defn}[thm]{Definition}	
\theoremstyle{remark}

\def\beq{\begin{eqnarray}}
\def\eeq{\end{eqnarray}}
\def\bsp{\begin{split}}	
\def\esp{\end{split}}

\newcommand{\spc}[2]{ {\bar{#1}}^{\dot{#2}} }
\newcommand{\spcd}[2]{ {\bar{#1}}_{\dot{#2}} }

\newcommand{\be}{\begin{equation}}
\newcommand{\ee}{\end{equation}}

\newcommand{\z}{\zeta}
\newcommand{\bz}{\bar{\zeta}}

\begin{document}

\title{Spacetimes with all scalar curvature invariants in terms of a cosmological constant}

\begin{abstract} 
In this letter we provide an invariant characterization for all spacetimes with all polynomial scalar invariants constructed from the Riemann tensor and its covariant derivatives vanishing except those zeroth order curvature invariants  expressed as polynomials in $\Lambda$, the cosmological constant. Using this invariant description we provide explicit forms for the metric. 
\end{abstract} 

\maketitle

\begin{section}{Introduction}

Given a metric for a spacetime, one may construct scalar curvature invariants by contracting the curvature tensor with various copies of itself. Curvature invariants of order $n>1$ are produced by contracting polynomials of the curvature tensor with its covariant derivatives up to order $n$. Extending the arguments made in \cite{BicakPravda} the entire collection of spacetimes with vanishing scalar invariants were identified \cite{4DVSI} by using the GHP formalism and the boost-weight decomposition to define balanced scalars and balanced spinors which in turn produce tensors which vanish upon contraction. Spacetimes with this property are said to be $VSI$, these spacetimes are a subclass of the $CSI$ spacetimes in which all polynomial scalar curvature invariants are constant \cite{4DKundt}

It was noted that this approach could be extended to the case with non-zero cosmological constant, $\Lambda \neq 0$ \cite{4DVSI}. This produces a subclass of the $CSI$ spacetimes for which all scalar curvature invariants either vanish or are polynomials in terms of $\Lambda$, denoted as the $CSI_{\Lambda}$ spacetimes. These are of interest as they are a natural and simple step from the $VSI$ and $CSI$ spacetimes revealing the interconnection between the two. 

As an example consider the plane-fronted gravitational waves, which constitute the entirety of the Petrov type N $CSI_{\Lambda}$ spacetimes. These were originally derived by Kundt \cite{Kundt61} in 1961 with vanishing cosmological constant. At the time, this was a reasonable constraint as it produced the simplest pure radiation solutions admitting a twist-free and non-expanding null congruence. Although the plausibility of a non-vanishing cosmological constant had been considered in \cite{Schro}, it was not until the 1981 that the Petrov type N solutions with cosmological constant were determined \cite{GP, SGP} and the plane-fronted gravitational waves in spacetimes with $\Lambda \neq 0$ were identified \cite{Ozvath}. 

The resulting class of $KN(\Lambda)[\alpha, \beta]$ metrics were classified by the sign of the cosmological constant $\Lambda \neq 0$ and another invariant $\kappa' = \frac13 \Lambda \alpha^2 + 2 \beta \bar{\beta}$ arising from the metric,
\beq ds^2 &=&-2q^2p^{-2} du\left( \left( -\frac{\kappa'}{2}v^2 + (ln q)_{,u} v + S(u,\z,\bz) \right) du + dv \right) + 2p^{-2}d\z d\bz ,~~ \nonumber \\
p &=& 1 + \frac{\Lambda}{6} \z \bz, \label{ORR} \\
q &=& (1-\frac{\Lambda}{6} \z \bz)\alpha(u) + \z \bar{\beta}(u) + \z \beta(u). \nonumber \eeq  
\noindent Excluding, the $\Lambda =0$ cases, this produces four canonical classes, which were shown to have a canonical form by setting $\alpha$ and $\beta$ to specific values using the appropriate coordinate transformations \cite{Bicak}:
\begin{itemize}
 \item $\kappa ' >0$, $\Lambda >0:$ $KN(\Lambda^{+})[0,1]$ 
 \item $\kappa ' >0$, $\Lambda <0:$ $KN(\Lambda^{-})[0,1]$
 \item $\kappa ' <0$, $\Lambda <0:$ $KN(\Lambda^{-})[1,0]$
 \item $\kappa ' =0$, $\Lambda <0:$ $KN(\Lambda^{-})[1,\sqrt{\frac{-\Lambda}{6}}e^{i\omega(u)}]$.
\end{itemize}
\noindent The physical interpretation of each of these subclasses is examined in \cite{BicakPodolsky} using the equations of geodesic deviation relative to an arbitrary timelike geodesic;  these may be interpreted as exact transverse gravitational waves with two polarization modes propagating on either Minkowski, de Sitter or anti-de Sitter space. 

The plane-fronted  gravitational waves, with $\Lambda = 0$, belong to the $VSI$ class of spacetimes \cite{4DVSI}, by adding a non-vanishing cosmological constant we have produced four distinct classes of $CSI$ spacetimes. It is reasonable to ask how many new distinct $CSI$ spacetimes are produced by adding $\Lambda \neq 0$ to each of the $VSI$ spacetimes. In light of the results of  \cite{4DKundt} we may classify the above solutions by examining the Segre type and comparing to the metric forms in \cite{4DKundt}. 

These spacetimes are of interest in quantum gravity. In the case of the vacuum plane wave spacetimes, it was shown that the vanishing of all scalar curvature invariants lead to all quantum corrections vanishing \cite{Deser72, Gibbons75}. Spacetimes for which all quantum corrections are a multiple of the metric are called {\it universal} \cite{Gibbons2008}; such spacetimes are important as they are solutions to the quantum theory, despite our lack of knowledge of the particular theory. Recently it was proven that any universal spacetime in four dimensions must have constant scalar curvature invariants \cite{ColeyHervik}. The $CSI_{\Lambda}$ spacetimes are a special subcase of the $CSI$ universal spacetimes. 

The goal of this letter will be to derive necessary and sufficient conditions on the Newman-Penrose scalar \cite{PR} for the class of $CSI$ spacetimes with all non-zero scalar curvature invariants expressed in terms of the cosmological constant $\Lambda \neq 0$, as a parallel to the result in \cite{4DVSI}. In the following section we state the theorem and break up the proof of the necessity and sufficiency of the conditions into two subsections. The third section employs the invariant characterization of the $CSI_{\Lambda}$ spacetimes along with the exhaustive list of $CSI$ spacetimes to identify all of the metrics for the $CSI_{\Lambda}$ spacetimes. Finally in the last section we discuss the relevance of the $CSI_{\Lambda}$ spacetimes to the equivalence problem for Lorentzian manifolds.

\end{section}
\begin{section}{The $CSI_{\Lambda}$ Theorem}
We wish to provide a simple set of conditions for spacetimes in which the Ricci Scalar is constant, and the only curvature invariants which are non-zero are the zeroth order invariants expressed as various polynomials of the cosmological constant $\Lambda$. 
\begin{thm}\label{thm:1}
Given a spacetime, all invariants constructed from the traceless Ricci tensor, Weyl tensor and their covariant derivatives vanish, if and only if the following conditions are satisfied:
\begin{enumerate}
\item The spacetime possesses a non-diverging, shear-free geodesic null congruence. 
\item Relative to this congruence, the Ricci Scalar is constant and all other Newman Penrose curvature scalars \cite{PR} with non-negative boost-weight vanish.
\end{enumerate}
\noindent These spacetimes belong to the $CSI$ class of spacetimes and we will say they are $CSI_{\Lambda}$ spacetimes.
\end{thm}
\noindent We choose the tangent vector to the null congruence to be $\ell^a$ and a spin basis so that $o^A \bar{o}^{\dot{A}} \leftrightarrow \ell^a$. The analytic conditions of theorem \ref{thm:1} (1) for this spin basis is expressed in terms of the vanishing of the spin coefficients \cite{PR}:
\beq \kappa = \rho = \sigma = 0 \label{krs}, \eeq
\noindent and the second condition of theorem \ref{thm:1} may be expressed as
\beq & \Psi_0 = \Psi_1 = \Psi_2 = 0, & \label{psi0} \\
& \Phi_{00} = \Phi_{01} = \Phi_{02} = \Phi_{11} = 0 &\label{phi0} \\
& \Lambda \equiv constant & \label{LambdaC} \eeq
Following the work done for $VSI$ spacetimes, the definitions and results given in the Necessity and Sufficiency proof of \cite{4DVSI} may be used to generalize the case where $\Lambda \neq 0$ is constant.
\begin{subsection}{Sufficiency of the Conditions}
To prove this direction of theorem \ref{thm:1} we will use the Newmann-Penrose (NP) and the compacted (GHP) formalisms given in section 4.12 of \cite{PR}; to start we introduce briefly the NP formalism before introducting the derivative operators of the GHP formalism. Throughout this paper we use a normalized spin basis $\{o^A, \iota^A \}$ such that $o^A \iota_A = 1$ and $o^A o_A = \iota^A \iota_A = 0$. From this we may build the corresponding tetrad:
 \beq & \ell^a \leftrightarrow o^A \spc{o}{A},~~n^a \leftrightarrow \iota^A \spc{\iota}{A},~~m^a \leftrightarrow o^A \spc{\iota}{A},~~\bar{m}^a \leftrightarrow \iota^A \spc{o}{A}, & \eeq
\noindent with the usual non-zero scalar products $-\ell_a n^a = m^a \bar{m}_a =1$. 
The spinorial form of the Riemann tensor $R_{abcd}$ is 
\beq  R_{abcd} & \leftrightarrow & \chi_{ABCD}\bar{\epsilon}_{\dot{A}\dot{B}}\bar{\epsilon}_{\dot{C}\dot{D}} + \bar{\chi}_{\dot{A} \dot{B} \dot{C} \dot{D} } \epsilon_{AB} \epsilon_{CD} \nonumber \\
&+& \Phi_{AB\dot{C}\dot{D}}\bar{\epsilon}_{\dot{A}\dot{B}}\epsilon_{CD} + \bar{\Phi}_{\dot{A} \dot{B}     CD}\epsilon_{AB}\bar{\epsilon}_{\dot{C} \dot{D}} \label{RiemTensor} \eeq
\noindent where
\beq \chi_{ABCD} = \Psi_{ABCD} + \Lambda(\epsilon_{AC} \epsilon_{BD} + \epsilon_{AD} \epsilon_{BD}) \label{Riemann} \eeq
\noindent and $\Lambda = R/24$ with $R$ the Ricci scalar. The Weyl spinor $\Psi_{ABCD} = \Psi_{(ABCD)}$ is related to the Weyl tensor by
\beq C_{abcd} = \Psi_{ABCD}\bar{\epsilon}_{\dot{A}\dot{B}}\bar{\epsilon}_{\dot{C}\dot{D}} + \bar{\Psi}_{\dot{A} \dot{B} \dot{C} \dot{D} } \epsilon_{AB} \epsilon_{CD}. \label{Weyl} \eeq 
\noindent Taking projections of this tensor onto the basis spinors $o^A, \iota^A$ give five complex scalar quantities $\Psi_i, i\in [0,4]$. Similarly the Ricci Spinor $\Phi_{AB\dot{C}\dot{D}} = \Phi_{(AB)(\dot{C}\dot{D})} = \bar{\Phi}_{\dot{A} \dot{B} CD}$ is connected to the traceless Ricci tensor $S_{ab} = R_{ab} - \frac14 R g_{ab}$ 
\beq \Phi_{AB\dot{C}\dot{D}} \leftrightarrow - \frac12 S_{ab}. \label{Ricci} \eeq 
\noindent We denote the projections of  $\Phi_{AB\dot{C}\dot{D}}$ onto $o^A, \iota^A$ by $\Phi_{00} = \bar{\Phi}_{00}$, $\Phi_{01} = \bar{\Phi}_{10}$, $\Phi_{02} = \bar{\Phi}_{20}$, $\Phi_{11} = \bar{\Phi}_{11}$, $\Phi_{12} = \bar{\Phi}_{21}$ and $\Phi_{22} = \bar{\Phi}_{22}$

The analytic expressions of theorem \ref{thm:1} $(1),(2)$ imply 
\beq \Psi_{ABCD} &=& \Psi_4 o_A o_B o_C o_D - 4 \Psi_3 o_{(A}o_B o_C \iota_{D)}, \label{myweyl} \\
\Phi_{AB \dot{C} \dot{D}} &=& \Phi_{22} o_A o_B \spcd{o}{C} \spcd{o}{D} -2 \Phi_{12} \iota_{(A}o_{B)} \spcd{o}{(C} \spcd{o}{D)} -2 \Phi_{21} o_{(A}o_{B)} \spcd{\iota}{(C} \spcd{o}{D)}. \label{myricci} \eeq
\noindent Following the convention used in \cite{4DVSI} we will say a scalar $\eta$ is a weighted quantity of type $\{p,q \}$ if for every non-vanishing scalar field $\lambda$, a transformation of the form
\beq o^A \mapsto \lambda o^A , \iota^A \mapsto \lambda^{-1} \iota^A, \nonumber \eeq
\noindent representing a boost in the $\ell^a - n^a$ plane and a spatial rotation in the $m^a - \bar{m}^a$ plane transformations $\eta$ in the following manner
\beq \lambda^p \bar{\lambda}^q \eta \nonumber \eeq
\noindent The boost weight, b, of a weighted quantity is defined by $b = \frac12 (p+q)$. 

The frame derivatives are defined as 
\beq & D = \ell^a \nabla_a = o^A \spc{o}{A} \nabla_{A \dot{A}}, \delta = m^a \nabla_a = o^A \spc{\iota}{A} \nabla_{A \dot{A}} & \nonumber \\
& D' = n^a \nabla_a = \iota^A \spc{\iota}{A} \nabla_{A \dot{A}}, \delta' = \bar{m}^a \nabla_a = \iota^A \spc{o}{A} \nabla_{A \dot{A}} & \nonumber \eeq
\noindent and so the covariant derivative may be expressed in terms of the frame,
\beq \nabla^a = \nabla^{A \dot{A}} = \iota^A \spc{\iota}{A} D + o^A \spc{o}{A} D' - \iota^A \spc{o}{A} \delta - o^A \spc{\iota}{A} \delta'. \nonumber \eeq
\noindent The GHP formalism introduces  new derivative operators $\eth, \th , \eth'$ and $\th'$ which are additive and obey the Leibniz rule. By including the spin-coefficient $\beta$ in the expression for these operators,  they act on scalars, spinors and tensors $\eta$ of type $\{p,q \}$ as follows:
\beq & \th = (D + p \gamma' + q \bar{\gamma}')\eta , \eth = (\delta + p \beta + q \bar{\beta}')\eta & \label{GHPderiv} \\
& \th' = (D' - p \gamma - q \bar{\gamma})\eta , \eth' = (\delta' + p \beta' + q \bar{\beta})\eta. & \nonumber  \eeq

To show the sufficiency conditions we assume the analytic conditions of theorem \ref{thm:1} hold along with the requirement that $o^A, \iota^A$ are parallely propogated along $\ell^a$ as well. Due to \eqref{krs} we have the following relations on the spin coefficients
\beq \gamma' = \tau' = 0. \label{sufcc} \eeq 
\noindent The spin-coefficient equations, the Bianchi identities and commutator relations \cite{PR} are greatly simplified by imposing \eqref{phi0}, \eqref{psi0}, \eqref{LambdaC}. The non-trivial relations that apply to proving the theorem are:
\beq \th \tau &=& 0, \label{Sa} \\
\th \sigma' &=& 0, \label{Sb} \\
\th \rho' &=& - 2 \Lambda, \label{Sc} \\
\th \kappa' &=& \tau \th' + \tau \sigma' - \Psi_3 - \Phi_{21}, \label{Sd} \\
\th \Psi_3 &=& 0, \label{Se} \\
\th \Phi_{21} &=& 0, \label{Sf} \\
\th \Phi_{22} &=& \eth' \Phi_{21} + (\eth' - 2\tau)\Psi_3 , \label{Sg} \\
\th \Psi_4 &=& \eth' \Psi_3 + (\eth' - 2\bar{\tau})\Phi_{21}, \label{Sh} \\
\th \th' - \th' \th &=& \bar{\tau} \eth + \tau \eth' + p \Lambda + q \Lambda, \label{Si} \\ 
\th \eth - \eth \th &=& 0. \label{Sj} \eeq

To proceed we analyze the boost weights of the quantities involved in these relations. In particular we will use the idea of a balanced scalar. 
\begin{defn} \label{defn:bs} 
Given a weighted scalar $\eta$ with boost-weight $b$, we shall say it is balanced if $\th^{-b} \eta = 0$  for~ $b<0$ 
and $\eta = 0$ for $b \geq 0$. 
\end{defn}
\noindent  Many of the lemmas as given in \cite{4DVSI} follow without change despite $\Lambda$'s non-vanishing. The proof of lemma 4 requires some modification due to \eqref{Sc}. For that reason, we will state each lemma leading to the main result without proof, unless there is some required change due to $\Lambda \neq 0$:
\begin{lem} \label{lem:vsi3}
 If $\eta$ is a balanced scalar then $\bar{\eta}$ is also balanced. 
\end{lem}
\begin{lem} \label{lem:vsi4}
If $\eta$ is a balanced scalar then, \beq  & \tau \eta,~\rho ' \eta,~\sigma ' \eta,~\kappa '\eta & \nonumber \\ & \th \eta,~ \eth \eta,~ \eth ' \eta,~\th ' \eta & \nonumber \eeq
\noindent are all balanced as well.
\end{lem}

\begin{table}[b] 
\begin{center} 
\begin{tabular}{c|c| c| c|}
&p&q & b  \\ [0.5 ex] \hline
$o^A$ & 1 & 0 & $\frac12$  \\ [0.5ex] \hline 
$\kappa$ & 3 & 1 & 2 \\ [0.5 ex]
$\sigma$ & 3 & -1 & 1 \\ [0.5 ex]
$\rho$ & 1 & 1 & 1 \\ [0.5 ex]
$\tau$ & 1 & -1 & 0 \\ [0.5 ex] \hline
$\TH$ & 1 & 1 & 1 \\ [0.5 ex]
$\eth$ & 1 & -1 & 0 \\ [0.5 ex] \hline
$\Psi_r$ & 4-2r & 0 & 2-r \\ [0.5 ex] 
 & & & \\ [0.5 ex]
\end{tabular}
\begin{tabular}{c|c| c| c}
&p&q & b  \\[0.5 ex] \hline
$\iota^A$ & -1 & 0 & $-\frac12$  \\ [0.5ex] \hline 
$\kappa'$ & -3 & -1 & -2 \\ [0.5 ex]
$\sigma'$ & -3 & 1 & -1 \\ [0.5 ex]
$\rho'$ & -1 & -1 & -1 \\ [0.5 ex]
$\tau'$ & -1 & 1 & 0 \\ [0.5 ex] \hline
$\TH'$ & -1 & -1 & -1 \\ [0.5 ex]
$\eth'$ & -1 & 1 & 0 \\ [0.5 ex] \hline
$\Phi_{rt}$ & 2-2r & 2-2t & 2-r-t \\ [0.5 ex] 
$\Lambda$ & 0 & 0 & 0 \\ [0.5 ex]
\end{tabular} \caption{Boost weights of weighted quantities}
\label{1table0} 
\end{center}
\end{table}
\begin{proof}
Let $b$ be the boost-weight of a balanced scalar $\eta$. Using table \ref{1table0} it is clear that the scalars listed in the first row have boost-weights $b,b-1,b-1,b-2$, respectively.To show these are balanced we must prove that the following must vanish: \beq & \th^{-b} (\tau \eta),~\th^{1-b}(\rho ' \eta), \th^{1-b}(\sigma ' \eta),~ \th^{2-b}(\kappa ' \eta).& \nonumber \eeq
\noindent while for the second row we require that four more quantities vanish to match their boost-weight:
\beq & \th^{-(b+1)}(\th \eta),~ \th^{-b}(\eth \eta),~ \th^{-b}(\eth ' \eta),~ \th^{-(b-1)}(\th ' \eta). & \nonumber \eeq
As the equations \eqref{Sc} and \eqref{Si} are the only that differ from the $VSI$ case, we must only check to see if two conditions still hold  \beq \th^{1-b}(\rho ' \eta) = \th^{1-b}(\th ' \eta) = 0 \nonumber \eeq
\noindent and the remaining six conditions hold automatically. The first condition follows using the Leibniz rule and equations \eqref{Sc} and the fact that $\th^2 \rho ' =0$.  Since we may expand this as
\beq & \th^{1-b}(\rho ' \eta) = \th \rho ' \th^{-b} \eta + \rho ' \th(\th^{-b} \eta), & \nonumber \eeq 
\noindent As $\eta$ is a balanced scalar for which $b<0$, these last two terms vanish. To prove the second condition, we use the commutator relation \eqref{Si} and the constancy of $\Lambda$ to get \beq  \th^{1-b}(\th ' \eta) &=& \th^{-b} (\th ' \th \eta) + \bar{\tau}(\th^{-b} \eth \eta) + \tau(\th^{-b} \eth \eta) + \th^{-b}(p \Lambda \eta + q \Lambda \eta) \nonumber \\ &=& \th^{-b} (\th ' \th \eta) \nonumber \eeq     
\noindent Using induction one may show that $\th^{1-b} \th ' \eta = \th ' \th^{1-b} \eta = 0$.
\end{proof}
\begin{lem} \label{lem:vsi5}
 If $\eta_1, \eta_2$ are balanced scalars both of type $\{ p,q\}$ then $\eta_1 + \eta_2$ is a balanced scalar of type $\{ p,q \}$ as well.
\end{lem}
\begin{lem} \label{lem:vsi6} 
 If $\eta_1$ and $\eta_2$ are balanced scalars then $\eta_1 \eta_2$ is also balanced.
\end{lem}
\begin{defn} \label{def:vsi7}
 A balanced spinor is a weighted spinor of type $\{ 0,0 \}$ whose components are all balanced scalars. 
\end{defn}
\begin{lem} \label{lem:vsi8}
 If $S_1$ and $S_2$ are balanced spinors then $S_1 S_2$ is also a balanced spinor 
\end{lem}
\begin{lem} \label{lem:vsi9}
 A covariant derivative of an arbitrary order of a balanced sinpor S is again a balanced spinor
\end{lem}
\begin{proof}
 Applying the covariant derivative to a balanced spinor $S$, \beq \nabla^a = \nabla^{A \dot{A}} = \iota^A \spc{\iota}{A} D + o^A \spc{o}{A} D' - \iota^A \spc{o}{A} \delta - o^A \spc{\iota}{A} \delta'. \nonumber \eeq
\noindent From table 1 in \cite{4DVSI} it follows that $\nabla^{A\dot{A}}S$ is a weighted spinor of type $\{ 0,0 \}$. By virtue of how $\th, \eth, \th'$ and $\eth'$ act on the basis vectors, the components may be shown to be balanced scalars using lemmas \ref{lem:vsi3},  \ref{lem:vsi4} and \ref{lem:vsi5}.	
\end{proof}
\begin{lem} \label{lem:vsi10} 
 A scalar constructed as a contraction of a balanced spinor is equal to zero. 
\end{lem}
From table 1 in \cite{4DVSI}, and equations \eqref{Se}- \eqref{Sh} it follows that the Weyl spinor and Ricci spinor and their complex conjugates are balanced spinors (lemma \ref{lem:vsi3}). Their product and covariant derivatives of arbitrary orders are balanced spinors as well (lemmas \ref{lem:vsi8} and \ref{lem:vsi9}). At this point to prove the sufficiency of the conditions of  theorem \ref{thm:1} we must state two more results:
\begin{lem} \label{lem:CSIL0}
The product of a balanced spinor and a weighted constant of type $\{0,0 \}$ is a balanced spinor.
\end{lem}
\begin{lem} \label{lem:CSIL1}
A scalar constructed as a contraction from the product of a balanced spinor, $\epsilon_{AB}$, $\epsilon^{AB}$ and their conjugates is equal to zero.
\end{lem}
With these observations, and equations \eqref{RiemTensor}, \eqref{Riemann}, \eqref{Weyl} and \eqref{Ricci}  imply that any contraction of the product of $N$ copies of the Riemann tensor with itself must vanish except for the contraction of the term built exclusively out of the product of $N$ copies of \beq \Lambda(\epsilon_{AC} \epsilon_{BD} + \epsilon_{AD} \epsilon_{BD} \bar{\epsilon}_{\dot{A} \dot{B}}\bar{\epsilon}_{\dot{C}\dot{D}}). \nonumber \eeq
\noindent Lemma \ref{lem:CSIL0} ensures all other terms are the products of balanced spinors, $\epsilon$'s and $\bar{\epsilon}$'s; these terms must vanish when contracted by lemmas \ref{lem:vsi10} and \ref{lem:CSIL1}.  To show that all non-zero curvature invariants appear at zeroth order, we note that the $n^{th}$ covariant derivative of the Riemann tensor is a balanced spinor for $n>0$, as $\nabla \epsilon_{AB} = 0$ and $\Lambda$ is a constant. Thus any product of the Riemann tensor with its $n^{th}$ covariant derivative must vanish upon contraction by lemma \ref{lem:CSIL1}, while any contraction of the product of the $n^{th}$ and $m^{th}$ covariant derivative of the Riemann tensor must vanish necessarily by lemma \ref{lem:vsi10}.  
\end{subsection}
\begin{subsection}{Necessity of the Conditions}
To show that these conditions are necessary follows by repeating the proof from \cite{4DVSI} verbatim, this can be done because the particular Newman Penrose equations used and the Bianchi Identities do not involve $\Lambda$, or the derivatives of $\Lambda$ - since they vanish if $\Lambda$ is constant. Requiring that all invariants vanish except those constructed as polynomials of $\Lambda$ which are assumed to be constant, one may prove conditions $(1)$ and $(2)$ of theorem \ref{thm:1} hold. 
\end{subsection}
\end{section}

\begin{section}{Local description of all $CSI_\Lambda$ spacetimes}
To the author's knowledge, an explicit list of metrics has yet to be given for the $CSI_{\Lambda}$  spacetimes However, portions of the $CSI_{\Lambda}$ spacetimes have been studied under other pretenses. For example, the whole of the $CSI_{\Lambda}$ spacetimes with $\Lambda = 0$ are known; these are the $VSI$ spacetimes \cite{4DVSI}. 

For non-zero $\Lambda$ we list known examples by Petrov Type.  The plane-fronted gravitational waves constitute all of the Petrov type N $CSI_\Lambda$ spacetimes. Metrics for these spacetimes were found in \cite{Ozvath} and these were further classified into canonical forms in \cite{Bicak}. As a generalization of the Kundt waves, all Petrov Type III $CSI_\Lambda$ spacetimes admitting pure radiation ($\Phi_{12} = 0$) were listed in \cite{GDP2004}. Lastly in the case of Petrov Type O, all spacetimes with Segre type $\{(2,11)\}$ \cite{Invskea} have been invariantly classified, implying that the $CSI_{\Lambda}$ spacetimes in this subcase are all known. 

We concentrate on the complete list of all CSI spacetimes given in \cite{4DKundt}, by imposing the necessary conditions on the curvature scalars \cite{PR} we can identify those $CSI$ metrics which belong to the $CSI_{\Lambda}$ case. We may interpret the vanishing or non-vanishing of $\Phi_{22}, \Phi_{02}$ and $\Phi_{20}$ in terms of the Segre type, \cite{Zakhary}, to determine the metric forms permitted for the $CSI_\Lambda$ spacetimes in \cite{4DKundt}: $\{(1,111)\}, \{(2,11)\}$ and $ \{(3,1)\}$. Employing Kundt coordinates, any $CSI$ spacetime may be expressed as
\beq ds^2 = 2du[dv+H(v,u,x^k)du+W_i(v,u,x^k) dx^i] + g_{ij}(x^k)dx^i dx^j \nonumber \eeq
\noindent where $dS^2_H = g_{ij}(x^k)dx^idx^j$ is the locally homogeneous metric of the transverse space and the metric functions $H$ and $W_i$ are functions of the form
\beq W_i(v,u,x^k) &=& vW_i^{(1)}(u,x^k)+W_i^{(0)}(u,x^k), \nonumber \\
H(v,u,x^k) & = & v^2 \tilde{\sigma} + vH^{(1)}(u,x^k)+H^{(0)}(u,x^k), \label{CSImf} \\
\tilde{\sigma} &=& \frac18 (4\sigma + W^{i(1)}W^{(1)}_i), \nonumber \eeq
\noindent where $\sigma$ is a constant. 

As the transverse space must be a locally homogeneous two-dimensional space, up to scaling, there are (locally) only the sphere $S^2$, flat space and the Hyperbolic plane $\mathbb{H}^2$. Exploiting this fact we list all of the $CSI_\Lambda$ spacetimes by the transverse metric and the one-form ${\bf W}^{(1)} = W_i^{(1)}dx^i$. The constant $\sigma$ will be specified in each case. 

\begin{subsection}{The Sphere $S^2$}
For those metrics with Segre type  $\{(1,111)\}, \{(2,11)\}$ and $ \{(3,1)\}$, one must have $\sigma>0$ and the transverse metric expressed as 
\beq ds^2_S = dx^2 + \frac{1}{\sigma} sin^2(\sqrt{\sigma}x)dy^2, \nonumber \eeq
\noindent where
\begin{enumerate}
\item ${\bf W}^{(1)} = 2\sqrt{\sigma}tan(\sqrt{\sigma} x)dx$. 
\item ${\bf W}^{(1)} = 2\sqrt{\sigma}[cot(\sqrt{\sigma} x)dx+tan(\sqrt{\sigma}y)dy]$. 
\item ${\bf W}^{(1)} = 2\sqrt{\sigma}[cot(\sqrt{\sigma} x)dx-cot(\sqrt{\sigma}y)dy]$. 
\end{enumerate}
\noindent Depending on the form of $H^{(1)}$, $H^{(0)}$and $W^{(0)}$ these are of Petrov type III,N or O.
\end{subsection}
\begin{subsection}{The Euclidean plane $\mathbb{E}^2$}
For those metrics with Segre type  $\{(1,111)\}, \{(2,11)\}$ and $ \{(3,1)\}$, the transverse metric will be
\beq ds^2_S = dx^2 + dy^2, \nonumber \eeq
\noindent where the $\sigma$ and the one-form are now
\begin{enumerate}
\item $\sigma = 0$, ${\bf W}^{(1)} = \frac{2\epsilon}{x}dx$, where $\epsilon = 0,1$.  
\end{enumerate}
\noindent This is the VSI case and so these are of Petrov type III, N or O. 
\end{subsection}
\begin{subsection}{The Hyperbolic plane $\mathbb{H}^2$}
Segre type for these metrics are $\{(1,111)\}, \{(2,11)\}$ and $ \{(3,1)\}$. We require that $\sigma<0$ and set $\sigma = -q^2$, depending on the case we will use different coordinates for the Hyperbolic plane.
\begin{enumerate}
\item $ds^2 = dx^2 + e^{-2qx} dy^2,~~ {\bf W}^{(1)} = 2qdx+\frac{2\epsilon}{y}dy$, where $\epsilon = 0,1$.
\item $ds^2 = dx^2 + \frac{1}{q^2} sinh^2(qx) dy^2,~~ {\bf W}^{(1)} = -2qtanh(qx) dx$.
\item $ds^2 = dx^2 + \frac{1}{q^2} sinh^2(qx) dy^2,~~ {\bf W}^{(1)} = 2q[coth(qx)dx - tanh(qy) dy]$.
\item $ds^2 = dx^2 + \frac{1}{q^2} cosh^2(qx) dy^2,~~ {\bf W}^{(1)} = 2q coth(qx) dx$.
\item $ds^2 = dx^2 + \frac{1}{q^2} cosh^2(qx) dy^2,~~ {\bf W}^{(1)} = 2q[-tanh(qx)dx + coth(qy) dy]$.
\end{enumerate}
For all of these cases the Petrov type is III, N or O. 
\end{subsection}
\end{section}
\begin{section}{Discussion}
In this paper the results of \cite{4DVSI} were extended to produce an invariant characterization of the class of spacetimes whose non-zero scalar curvature invariants are expressed as polynomials of $\Lambda$ ($CSI_{\Lambda}$ spacetimes) by determining necessary and sufficient conditions on the Newman-Penrose curvature scalars \cite{PR}. Then by employing the exhaustive list of metrics for the $CSI$ spacetimes \cite{4DKundt} all of the $CSI_{\Lambda}$ spacetimes are found by comparing Segre Type, and the sign of the cosmological constant $\Lambda$. These invariants determine the metrics whose range of Petrov types: $III,N,$ and $O$ which complete the necessary and sufficient conditions for these spacetimes to be $CSI_{\Lambda}$.

In this sense all of the $CSI_{\Lambda}$ spacetimes have been given a local description in terms of a metric, however the conditions on the metrics for the subcases arising from the Petrov classification are not known. Equivalently, the interconnection between the Cartan invariants and the specialization in Petrov type $III \to N \to  O$ is unknown.  In order to answer such a question one must apply the Karlhede algorithm to the entirety of the $CSI_{\Lambda}$ metrics. Such a task which would be an effort to implement but entirely plausible; recently all vacuum Petrov type N $VSI$ spacetimes have been invariantly classified using the Cartan invariants arising from the application of the  Karlhede algorithm  to the vacuum PP-waves \cite{Milson}, and the vacuum Kundt waves \cite{ McNutt}.  

Such a classification would give insight into the application of the Karlhede algorithm to $CSI$ spacetimes, which is an important question related to the equivalence problem for these spacetimes. Furthermore the collection of $CSI$ metrics have higher-dimensional analogues  \cite{hdcoley, HDVSI, Kundt}.  By comparing the four dimensional subcases with their higher dimensional counterparts it is hoped an analogue of the Karlhede algorithm could be implemented for all $CSI_{\Lambda}$ spacetimes. 
\end{section}

\end{document}